\documentclass[12pt, twoside, letterpaper]{amsart}

\usepackage{amsmath, hyperref, color, amssymb}
\usepackage{srcltx}
\usepackage[normalem]{ulem}

\usepackage{geometry}
\usepackage{amsthm} \usepackage{amsmath} \usepackage{amsfonts}
\usepackage{amssymb} \usepackage{latexsym} \usepackage{enumerate}
\usepackage[numbers]{natbib}

\theoremstyle{plain}%
\newtheorem{Theorem}{Theorem}[section] %

\newtheorem{Proposition}[Theorem]{Proposition} %

\theoremstyle{definition}%

\newtheorem{definition}[Theorem]{Definition}
\newtheorem{interpretation}[Theorem]{Interpretation}%
\newtheorem{Example}[Theorem]{Example}
\newtheorem{remark}[Theorem]{Remark} %

\theoremstyle{remark}%
\newtheorem{Remark}[Theorem]{Remark} %

\newcommand{\Cov}[3]{\cov_{#1}\left[#2,#3\right]}
\newcommand{\Var}[2]{\var_{#1}
\left[#2\right]}

\DeclareMathOperator{\cov}{Cov}
\DeclareMathOperator{\var}{Var}

\newcommand{\envspace}{\vspace{2mm}}

\def\E{\mathbb{E}}
\def\F{\mathbb{F}}

\def\P{\mathbb{P}}

\def\R{\mathbb{R}}

\def\e{\varepsilon}

\def\vth{\vartheta}

\def\w{\omega}

\def\cF{\mathcal{F}}

\usepackage{subcaption}
\usepackage{tikz}
\usepackage{latexsym, amssymb, bbm}
\usepackage{amsmath}
\usepackage{mathrsfs}
\usepackage{amsthm}
\usepackage{verbatim}
\usepackage{fancyhdr}
\usepackage{tikz}
\usetikzlibrary{arrows}
\usepackage{caption}
\usepackage{forest}
\usepackage{graphicx}
\usepackage{stmaryrd}
\usepackage{wasysym}

\begin{document}
\title
[Stability and asymptotic analysis of the F\"ollmer-Schweizer decomposition]{Stability and asymptotic analysis of the F\"ollmer-Schweizer decomposition on a finite probability space}
\author{Sarah Boese}  
\address{Sarah Boese,  Vassar College, Poughkeepsie, NY 12604, United States}
\email{sboese@vassar.edu}

\author{Tracy Cui} 
\address{Tracy Cui, Carnegie Mellon University, Pittsburgh, PA 15213, United States}
\email{txxcui@gmail.com}

\author{Samuel Johnston} 
\address{Samuel Johnston, Willamette University, Salem, OR 97301, United States}
\email{sdjohnston@willamette.edu}

\author{Gianmarco Molino}
\address{Gianmarco Molino, Department of Mathematics, University of Connecticut, Storrs, CT 06269, United States}
\email{gianmarco.molino@uconn.edu}

 \author{Oleksii Mostovyi}
\address{Oleskii Mostovyi, Department of Mathematics, University of Connecticut, Storrs, CT 06269, United States}
\email{oleksii.mostovyi@uconn.edu}
\thanks{This paper is a part of an REU project conduced in  Summer 2019 at the University of Connecticut. The project has been supported by the National Science Foundation under grants  No. DMS-1659643 and  DMS-1848339.}
\keywords{F\"ollmer-Schweizer decomposition, simultaneous perturbations of the drift and volatility, asymptotic analysis, stability}
\subjclass[2010]{60G07, 93E20, 91G10, 91G20, 90C31}
\begin{abstract}
    First, we consider the problem of hedging in complete binomial models. Using the discrete-time F\"ollmer-Schweizer decomposition, we demonstrate the equivalence of the backward induction and sequential regression approaches. Second, in incomplete trinomial models, we examine the extension of the sequential regression approach for approximation of contingent claims. Then, on a finite probability space, we investigate stability of the discrete-time F\"ollmer-Schweizer decomposition with respect to  perturbations of the stock price dynamics and, finally, perform its asymptotic analysis under simultaneous perturbations of the drift and volatility of the underlying discounted stock price process, where we prove stability and obtain explicit formulas for the leading order correction terms. 
\end{abstract}

\maketitle

\section{Introduction}
In practice, financial models are not exact -- as in any field, modeling based on real data introduces some degree of error. Therefore, it is important to understand the effect error has on the calculations and assumptions we make on the model. In this paper, we focus on the stability and asymptotic analysis of the F\"ollmer-Schweizer decomposition, as among the pricing and hedging approaches, in incomplete markets, it gives the best approximation of a given contingent claim in the sense of the least-squares error, which is one of the most natural criteria used in practice. Further, in complete models, the F\"ollmer-Schweizer decomposition is consistent with the backward induction, which is another canonical method in financial mathematics.

Steven E. Shreve's \textit{Stochastic Calculus for Finance I: The Binomial Asset Pricing Model} \cite{shreve_stochastic_2012} introduces option pricing in a highly accessible manner. The text predominantly focuses on the binomial model, and in this paper, we go beyond it, as there are many models used in practice that are not binomial. As the most natural extension of a binomial model is a trinomial one, below, we also give it special consideration. Note that both binomial and trinomial models, despite their simplicity, are widely used in approximations of pricing and hedging in more advanced models, including the continuous-time ones, see, e.g., \cite{BrigoMercurio}.

In this paper, we extend the introduction to asset hedging given by Shreve to the strategy of sequential regression, keeping the discrete-time framework but allowing for consideration of other market models, including incomplete ones. In the complete case, we show that the strategy of sequential regression introduced by F\"ollmer and Schweizer \cite{follmer_hedging_1988} is equivalent to Shreve's recursive hedging formula. We then extend our discussion to the incomplete trinomial model, after which we show small market perturbations have a little effect, which we quantify, on hedging strategies and option pricing, and derive formulas for correction factors. 

The remainder of this paper is organized as follows. In Section \ref{secProblem}, we formulate the minimization problem and define the F\"olmer-Schweizer decomposition. Section \ref{secCompleteMarkets} contains its investigation in complete binomial markets, where we also prove the equivalence of the approach based on the F\"olmer-Schweizer decomposition to the backward induction. Section \ref{secIncompleteMarkets}   presents the discussion of the general incomplete case. In Section \ref{secStability}, we revisit the stability question in the context of perturbations of the model parameters, where we introduce a parametrization of perturbations that allows for simultaneous distortions of its drift and volatility of the underlying stock price process. We prove the stability of the F\"olmer-Schweizer decomposition under such perturbations. Finally, in Section \ref{secAsymptotics}, we obtain explicit formulas for the first-order correction terms of each component of the  decomposition under such perturbations, including the correction to the optimal trading strategy.

\section{The discrete-time F\"ollmer-Schweizer decomposition}\label{secProblem}
Let $(\Omega,\P)$ be a finite probability space, $N$ a fixed positive integer and $\F=(\mathcal F_n)_{n=0,1,\dots,N}$ a filtration, i.e., an increasing family of sub-algebras, each containing $\Omega$ and $\emptyset$. Assume that $\cF_0$ is trivial and $\mathcal F_N$ contains all subsets of $\Omega$. As we work on a finite probability space, without loss of generality, we suppose that $\mathbb P[\w]>0$ for every $\w\in\Omega$.  We suppose that there is a bank account, which we will use as a num\'eraire, and in particular, its price process equals to $1$ at all times.   Let $S=(S_n)_{n=0,1,...,N}$ be a real-valued, $\F$-adapted process, i.e., each $S_n$ is $\mathcal F_n$-measurable.  $S$ describes the discounted price process of a stock. We denote 
\begin{displaymath}
    \Delta S_n:=S_n-S_{n-1} ,\quad \text{for } n=1,...,N
\end{displaymath}
the increments of $S$. We call a process $\vartheta=(\vartheta_n)_{n=1,...,N}$ \textit{predictable} if $\vartheta_n$ is $\cF_{n-1}$-measurable for each $n$. Let $\Theta$ be the set of all predictable processes $\vartheta$, that financially correspond to self-financing trading strategies, in view of the presence of money market account, which we use as a num\'eraire. 

\begin{definition} For $\vartheta \in \Theta$, we define the process $G(\vartheta)$ as follows: $G(\vartheta)$ as 
\begin{displaymath}
    G_n(\vartheta):= \sum^n_{j=1}\vartheta_{j}\Delta S_j \hspace{0.5in} \text{for } n=0,1,...,N.
\end{displaymath} 
\end{definition}

For a given a random variable  $V_N$ and $c \in \mathbb R$, one can consider the  following problem posed in \cite{schweizer_variance-optimal_1995}  
\begin{equation}\label{problem}
\begin{split}
    & \text{minimize } \E\left[\left(V_N-c-G_N(\vartheta)  \right)^2\right] \ \text{ over all } \vartheta \in \Theta ~\text{and}~c\in\mathbb R.   
    \end{split} 
\end{equation}

\begin{interpretation}
As we view $S_n$ as the price at time $n$ of a risky financial asset, the process $\vartheta$ describes the trading strategy of some investor in the market, where $\vartheta_n$ is the number of shares of stock held between the times $n$ and $n+1$. Process $G(\vartheta)$ becomes the \textit{gains from the trade} process. We now interpret $V_N$ as a \textit{nontraded security} measured in the units of the bank account with maturity $N$ and $c$ as the \textit{initial capital}. Thus problem \eqref{problem} can be interpreted as finding a self-financing trading strategy that gives the best least-squares approximation of $V_N$. Mathematically, \eqref{problem} is also closely related to the problem considered in \cite{SchweizerApprox}, finding the best approximation of a random variable by a stochastic integral (plus a constant). Quadratic optimization problems of the form \eqref{problem} also appear in the asymptotic analysis of stochastic control problems with respect to perturbations of the initial data, where they govern the second-order correction terms, see \cite{KS2006}, \cite{KS2006b}, \cite{MostovyiSirbuModel}, and \cite{MostovyiNumeraire} for details. 
\end{interpretation}

A solution \eqref{problem}, given in terms of an explicit formula for an optimal trading strategy $\widehat\vartheta$, is known as \textit{sequential regression} and is shown in \cite{follmer_hedging_1988}. 
For a general probability space, such a solution is subject to additional conditions on $S$ and is closely related to the discrete-time \textit{F\"ollmer-Schweizer decomposition}, defined below. 

\begin{definition} 
We use the definition of the \textit{nondegeneracy condition} (ND) as given in \cite{schweizer_variance-optimal_1995}, that is, $S$ satisfies (ND) if there exists a constant $\delta \in (0,1)$ such that 
\begin{align*}
    \left(\E[\Delta S_n \mid \cF_{n-1}]\right)^2 \leq \delta \E\left[\Delta S_n^2 \mid \cF_{n-1}\right] \hspace{0.5in} \P\text{-a.s. for } n=1,\dots,N.
\end{align*}
\end{definition}

\begin{remark}
We note that on finite probability spaces, (ND)  holds in non-trivial (or rather non-degenerate) cases.
\end{remark}
\begin{definition}\label{defFS}
We now introduce the \textit{discrete F\"ollmer-Schweizer decomposition}, following \cite{schweizer_variance-optimal_1995}. Let $$S=M+A,$$ be the semimartingale decomposition of $S$ into a martingale $M$ and a predictable process $A$. Random variable $V_N$ admits the \textit{discrete F\"ollmer-Schweizer decomposition} if it can be written as
\begin{equation} 
    V_N = V_0 + \sum_{j=1}^N\widehat \vartheta_j \Delta S_j+L_N.
    \label{decomp}
\end{equation} 
for some $V_0 \in \R$, a process $\widehat \vartheta\in \Theta$, and a $\mathbb P$-martingale $L$, such that
\begin{enumerate}
    \item $L$ and $M$ are orthogonal, i.e., 
    $LM$ is a $\mathbb P$-martingale,  
    \text{and}
    \item $\E[L_0]=0$.
\end{enumerate}
Note that when $\cF_0$ is trivial, the latter condition reads $L_0=0$.
\end{definition}
%
Using the sequential regression approach, following \cite{follmer_hedging_1988}, we obtain the following formula for an optimal hedging strategy $\widehat \vartheta$:

\begin{equation}\label{sequential}
    \widehat \vartheta_{n} := \frac{
    \Cov{\cF_{n-1}}{ V_N - \sum\limits_{j = n + 1}^{N}\widehat \vartheta_j\Delta S_j}{ \Delta S_{n} }
    }{\Var{\cF_{n-1}}{\Delta S_{n} }},\quad n=1,\dots,N,
\end{equation}
where $\Cov{\cF_{n-1}}{\cdot}{\cdot}$ and $\Var{\cF_{n-1}}{\cdot}$ denote the conditional covariance and variance, respectively.
This demonstrates the richness of the FS-decomposition as an analytic tool. With very limited assumptions about  $V_N$, we are able to obtain an explicit formula for the optimal (in the sense of \eqref{problem}) hedging strategy. Furthermore, this hedging formula holds in both complete and incomplete markets, which are discussed in the following sections.

\section{Complete Markets}\label{secCompleteMarkets}
In the settings of the previous section, when the market is complete, every contingent claim  $V_N$ can be represented as
\begin{displaymath}
    V_N=V_0+\sum^{N}_{j=1}\vartheta_{j}\Delta S_{j},
\end{displaymath}
for some $\vartheta \in \Theta$ and $V_0\in\mathbb R$.
Note that this situation corresponds to $L=0$ in  Definition \ref{defFS}. Put differently, $V_N$ is now can be obtained by trading between the money market account and the stock.
\subsection{Binomial Asset Pricing Model}
We now introduce a simple framework for the problem following \cite{shreve_stochastic_2012}, however directly  using the bank account as a num\'eraire. Consider a \textit{binomial asset pricing model}, where at each time step $k$, $S_{k+1}$ can take one of two values: $uS_k$ or $dS_k$, with probabilities $p\in(0,1)$ and $q := 1-p$, and where $u>1$ and $d\in (0,1)$, known as the \textit{up factor} and \textit{down factor}, respectively, are fixed positive constants with $u>d$. The value at each time $k$ is determined by a (not necessarily fair) coin flip $\omega_k$, which can take either the value $H$ or $T$ and is independent of other coin tosses. 
Let $\F = (\mathcal F_n)_{n = 0,\dots,N}$ be the filtration, where each $\mathcal F_n$ contains information about the first $n$ coin tosses. 
  An example of a 2-period binomial model is shown in the figure below. 

\begin{figure}[h]
    \centering
    \begin{tikzpicture}[line cap=round,line join=round,>=triangle 45,x=2cm,y=2cm]
    \draw [line width=0.5pt] (-3.4641016151377544,2)-- (-2.598076211353316,2.5);
    \draw [line width=0.5pt] (-3.4641016151377544,2)-- (-2.598076211353316,1.5);
    \draw [line width=0.5pt] (-2.598076211353316,2.5)-- (-1.7320508075688772,3);
    \draw [line width=0.5pt] (-2.598076211353316,2.5)-- (-1.7320508075688772,2);
    \draw [line width=0.5pt] (-2.598076211353316,1.5)-- (-1.7320508075688772,2);
    \draw [line width=0.5pt] (-2.598076211353316,1.5)-- (-1.7320508075688772,1);
    \begin{scriptsize}
    \draw [fill=black] (-3.4641016151377544,2) circle (1.5pt);
    \draw[color=black] (-3.4802097299192036,2.2641580080798347)  [yshift=-5pt] node {$S$};
    \draw [fill=black] (-2.598076211353316,2.5) circle (1.5pt);
    \draw[color=black] (-2.598250447005697,2.762205603136873)  [yshift=-5pt] node {$uS$};
    \draw [fill=black] (-2.598076211353316,1.5) circle (1.5pt);
    \draw[color=black] (-2.598250447005697,1.7557344214591075)  [yshift=-5pt] node {$dS$};
    \draw [fill=black] (-1.7320508075688772,3) circle (1.5pt);
    \draw[color=black] (-1.7059151725285018,3.2498772066302233)  [yshift=-5pt] node {$u^2S$};
    \draw [fill=black] (-1.7320508075688772,2) circle (1.5pt);
    \draw[color=black] (-1.7266671556558786,2.2641580080798347)  [yshift=-5pt] node {$S$};
    \draw [fill=black] (-1.7320508075688772,1) circle (1.5pt);
    \draw[color=black] (-1.6851631894011252,1.2784388095294457)  [yshift=-5pt] node {$d^2S$};
    \end{scriptsize}
    \end{tikzpicture}
    \caption{Example of a 2-Period Binomial Model}
    \label{binomialfig}
\end{figure}
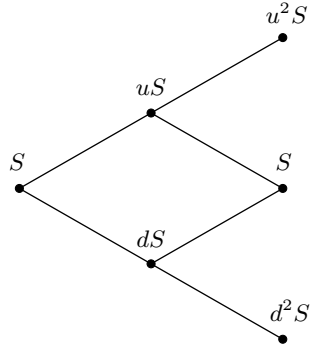

\begin{remark}
        For  figure \ref{binomialfig}, for the binomial model, we considered the case when  $d=\frac{1}{u}$. That is, after an even number of time steps, the stock price returns to its original value if we flip exactly the same number of heads and tails. However, in general, it is not necessarily the case that $u = \frac{1}{d}$.
    \end{remark}

\vspace{0.1in}

Let $X$ be the replicating wealth process for the contingent claim $V_N$, i.e., a self-financing process starting from the initial wealth $X_0$, and such that $$X_N = V_N.$$ Classical backward induction approach, for which we refer to \cite{shreve_stochastic_2012}, assures that the number of shares of stock in the replicating portfolio for $V_N$ held between times $n$ and $n+1$, can be obtained via the formula:
\begin{equation}\label{deltaHedging}
\vartheta_{n+1}(\w)=
\frac{X_{n+1}(\omega H) - X_{n+1}(\omega T)}{S_{n+1}(\omega H) - S_{n+1}(\omega T)},
\end{equation}
where $\omega = \omega_1,\dots,\omega_n$. 
This result is known as the Delta-Hedging, see \cite[Theorem 1.2.2, p. 12]{shreve_stochastic_2012}. 

The main result of this section is proving that in complete binomial settings the Delta-Hedging rule gives the same strategy as the F\"olmer-Schweizer decomposition.

\begin{Proposition}\label{propBinomial}
In complete binomial settings, the formulas \eqref{sequential} and \eqref{deltaHedging} are equivalent.
\end{Proposition}
\begin{proof}
Below $\widehat \vartheta$ will denote the strategy obtained from formula \eqref{sequential}, and $\vartheta$ will denote the replicating strategy given by \eqref{deltaHedging} for each $n$.  
First, for $n = N$,  \eqref{sequential} reads
\begin{equation}\label{sequentialN}
    \widehat \vartheta_{N} = \frac{
    \Cov{\cF_{N-1}}{ V_N}{ \Delta S_{N} }
    }{\Var{\cF_{N-1}}{\Delta S_{N} }}=\frac{
    \Cov{\cF_{N-1}}{ X_N}{ \Delta S_{N} }
    }{\Var{\cF_{N-1}}{\Delta S_{N} }},
\end{equation}
as $X_N = V_N$. Also, since\footnote{We use indices for $\vartheta$ in a consistent way with \cite{schweizer_variance-optimal_1995}, so that $\vartheta$ is predictable.} $$X_{n+1} = X_{n} + \vartheta_{n+1}(S_{n+1} - S_{n}), \quad n = 0,\dots,N-1,$$ 
we have
$$X_N - \mathbb E_{\cF_{N-1}}\left[X_N\right] = \vartheta_N \left(\Delta S_N - \mathbb E_{\cF_{N-1}}\left[ \Delta S_N\right] \right).$$
Therefore, we can rewrite $\Cov{\cF_{N-1}}{ V_N}{ \Delta S_{N} }$  as 
\begin{equation}\nonumber
\begin{split}
\Cov{\cF_{N-1}}{ V_N}{ \Delta S_{N} } &= 
\mathbb E_{\cF_{N-1}}\left[ (X_N - \mathbb E_{\cF_{N-1}}\left[X_N\right])(\Delta S_N - \mathbb E_{\cF_{N-1}}\left[ \Delta S_N \right])\right] \\
&= \vartheta_N
\mathbb E_{\cF_{N-1}}\left[ (\Delta S_{N} -  \mathbb E_{\cF_{N-1}}\left[\Delta S_{N} \right])(\Delta S_N - \mathbb E_{\cF_{N-1}}\left[ \Delta S_N \right])\right] \\
&= \vartheta_N
\Var{\cF_{N-1}}{\Delta S_{N}}.
\end{split}
\end{equation}
And thus $\vartheta_N = \frac{\Cov{\cF_{N-1}}{ V_N}{ \Delta S_{N} }}{\Var{\cF_{N-1}}{\Delta S_{N}}}$. Comparing $\vartheta_N$ to  $\widehat\vartheta_N$ from  \eqref{sequentialN}, we deduce that they coincide.  
This also implies that 
$$
X_{N-1} = X_N - \widehat\vartheta_N \Delta S_N.$$
Or equivalently, we have 
\begin{equation}\label{12251}X_N - \sum\limits_{j=N}^N\widehat\vartheta_N \Delta S_N = X_{N-1}.\end{equation}
The latter expression, however, is exactly the term appearing in $\widehat \vartheta_{N-1}$ in \eqref{sequential}, which via \eqref{12251}, we can rewrite as
$$\widehat\vartheta_{N-1} = \frac{\Cov{\cF_{N-2}}{X_N - \sum\limits_{j=N}^N\widehat\vartheta_N \Delta S_{N-1}}{\Delta S_{N-1}}}{\Var{\cF_{N-2}}{\Delta S_{N-1}}} = \frac{\Cov{\cF_{N-2}}{X_{N-1}}{\Delta S_{N-1}}}{\Var{\cF_{N-2}}{\Delta S_{N-1}}},$$
which differs from \eqref{sequentialN} only by the value of the index. Therefore, line by line applying the argument above used for proving that $\vartheta_N=\widehat\vartheta_N$, we can show that $\vartheta_{N-1}=\widehat\vartheta_{N-1}.$ Proceeding in such a way, we can show that $\vartheta_{n}=\widehat\vartheta_{n}$ for each $n\in\{0,\dots,N\}$.
\end{proof}

\begin{remark}
Notice that Proposition \ref{propBinomial} only demonstrates the equivalence of the specific strategies of backward recursion and sequential regression in the binomial model, not necessarily the uniqueness of the minimizing strategy. However, the uniqueness of the solution to \eqref{problem}, that is of the optimal stochastic integral, the associated strategy, and a constant follows  from the strict convexity of the quadratic objective in \eqref{problem} and some computations under the risk-neutral measure.
\end{remark}

\begin{Example}
Consider a 3-step binomial asset pricing model with $S_0=4,u=2,d=\frac{1}{2},r=\frac{1}{4},p=\frac{1}{2},q=\frac{1}{2}$, and a \textit{European Call Option} expiring at time $N=3$ with strike price $K=1$. Note that in this market, the  one-step risk-neutral probabilities are $\Tilde{p}=\Tilde{q}=\frac{1}{2}$ and the non-zero interest rate can be handled by considering  the discounted stock. We will illustrate the optimal hedge using both \eqref{sequential} and \eqref{deltaHedging}. Recall that the value at time $N$ of a European Call is given by 
\begin{align}
    V_N = (S_N-K)^+.
\end{align}
We compute the stock prices and discounted asset values at each time step for each possible combination of coin flips using the following formulas:
\begin{displaymath}
\begin{split}
    S_n(\w_1,...,\w_n) = S_0*u^{\text{(\# heads)}}*d^{\text{(\# tails)}} \\
    V_n(\w_1,...,\w_n) = \left(\frac{1}{1+r}\right) \Tilde{\E}[V_{n+1} \mid \w_1,...,\w_n],
    \end{split}
\end{displaymath}
where $\mathbb P$ is the risk-neutral probability measure. 
The trees of stock prices and discounted asset values are shown in Figure \ref{binomial}. For brevity, we will calculate a hedge using both formulas at time 2, since $A_N=V_N$ in the sequential regression formula at the penultimate time step. We aim to show that $\vartheta_2(TT)=\xi_2(TT)$. Calculating $\vartheta_2(TT)$ using \eqref{deltaHedging} yields 
\begin{displaymath}
    \frac{1-0}{2-\frac{1}{2}}=\frac{1}{\frac{3}{2}}=\frac{2}{3}.
\end{displaymath}

Using \eqref{sequential}, $\xi_2(TT)$ is given by 
\begin{displaymath}
\begin{split}
    & \frac{\Cov{\cF_2}{V_3}{\Delta S_3\mid TT}}{\Var{\cF_2}{(\Delta S_3\mid TT)}} \\
    & = \frac{\E_{\cF_{2}}[(V_3-\E_{\cF_{2}}[V_3 \mid TT])(\Delta S_3 - \E_{\cF_{2}}[\Delta S_3 \mid TT]) \mid TT]}{\E_{\cF_{2}}[(\Delta S_3 - \E_{\cF_{2}}[\Delta S_3 \mid TT])^2 \mid TT]} \\
    & = \frac{\E_{\cF_{2}}[(V_3-(\frac{1}{2}*1+\frac{1}{2}*0))(\Delta S_3 - (\frac{1}{2}*1+\frac{1}{2}*-\frac{1}{2})) \mid TT]}{\E_{\cF_{2}}[\Delta S_3^2 \mid TT]-\E_{\cF_{2}}[\Delta S_3 \mid TT]^2} \\
    & = \frac{\E_{\cF_{2}}[(V_3-\frac{1}{2})(\Delta S_3-\frac{1}{4}) \mid TT]}{\frac{5}{8}-\frac{1}{16}} \\
    & = \frac{(\frac{1}{2})(1-\frac{1}{2})(1-\frac{1}{4})+(\frac{1}{2})(0-\frac{1}{2})(-\frac{1}{2}-\frac{1}{4})}{\frac{9}{16}} \\
    & = \frac{\frac{3}{16}+\frac{3}{16}}{\frac{9}{16}} \\
    & = \frac{2}{3}.
    \end{split}
\end{displaymath}

This illustrates the equivalence of \eqref{sequential} and \eqref{deltaHedging} in the context of this example. \\

\begin{figure}[h]
\centering
\begin{subfigure}{.5\textwidth}
  \centering
  \begin{tikzpicture}[line cap=round,line join=round,>=triangle 45,x=1cm,y=1cm]
        \draw [-,line width=0.5pt] (-3.4641016151377544,-1) -- (-2.598076211353316,-0.5);
        \draw [-,line width=0.5pt] (-2.598076211353316,-0.5) -- (-1.7320508075688772,0);
        \draw [-,line width=0.5pt] (-1.7320508075688772,0) -- (-0.8660254037844386,0.5);
        \draw [-,line width=0.5pt] (-1.7320508075688772,0) -- (-0.8660254037844386,-0.5);
        \draw [-,line width=0.5pt] (-2.598076211353316,-0.5) -- (-1.7320508075688772,-1);
        \draw [-,line width=0.5pt] (-1.7320508075688772,-1) -- (-0.8660254037844386,-0.5);
        \draw [-,line width=0.5pt] (-3.4641016151377544,-1) -- (-2.598076211353316,-1.5);
        \draw [-,line width=0.5pt] (-2.598076211353316,-1.5) -- (-1.7320508075688772,-1);
        \draw [-,line width=0.5pt] (-1.7320508075688772,-1) -- (-0.8660254037844386,-1.5);
        \draw [-,line width=0.5pt] (-2.598076211353316,-1.5) -- (-1.7320508075688772,-2);
        \draw [-,line width=0.5pt] (-1.7320508075688772,-2) -- (-0.8660254037844386,-1.5);
        \draw [-,line width=0.5pt] (-1.7320508075688772,-2) -- (-0.8660254037844386,-2.5);
        \begin{scriptsize}
        \draw [fill=black] (-3.4641016151377544,-1) circle (1.5pt);
        \draw[color=black] (-3.4762183455025912,-0.8860699289787186) node [yshift=1ex] {$4$};
        \draw [fill=black] (-2.598076211353316,-0.5) circle (1.5pt);
        \draw[color=black] (-2.5575235643431045,-0.41123891849179217) node [yshift=1ex]  {$8$};
        \draw [fill=black] (-1.7320508075688772,0) circle (1.5pt);
        \draw[color=black] (-1.6697960229979825,0.08939812517377155) node [yshift=1ex]  {$16$};
        \draw [fill=black] (-0.8660254037844386,0.5) circle (1.5pt);
        \draw[color=black] (-0.80271330819577,0.5900351688393352) node [yshift=1ex]  {$32$};
        \draw [fill=black] (-0.8660254037844386,-0.5) circle (1.5pt);
        \draw[color=black] (-0.8233581347386799,-0.41123891849179217) node [yshift=1ex]  {$8$};
        \draw [fill=black] (-1.7320508075688772,-1) circle (1.5pt);
        \draw[color=black] (-1.6904408495408922,-0.9118759621573559) node [yshift=1ex]  {$4$};
        \draw [fill=black] (-2.598076211353316,-1.5) circle (1.5pt);
        \draw[color=black] (-2.5575235643431045,-1.407351799187192) node [yshift=1ex]  {$2$};
        \draw [fill=black] (-0.8660254037844386,-1.5) circle (1.5pt);
        \draw[color=black] (-0.8233581347386799,-1.407351799187192) node [yshift=1ex]  {$2$};
        \draw [fill=black] (-1.7320508075688772,-2) circle (1.5pt);
        \draw[color=black] (-1.6904408495408922,-1.9079888428527558) node [yshift=1ex]  {$1$};
        \draw [fill=black] (-0.8660254037844386,-2.5) circle (1.5pt);
        \draw[color=black] (-0.6014262494023993,-2.382819853339682) node [xshift=-1.6ex] [yshift=1.5ex]  {$\frac{1}{2}$};
        \end{scriptsize}
    \end{tikzpicture}
  \caption{Stock Prices}
\end{subfigure}%
\begin{subfigure}{.5\textwidth}
  \centering
  \begin{tikzpicture}[line cap=round,line join=round,>=triangle 45,x=1cm,y=1cm]
        \draw [-,line width=0.5pt] (-3.4641016151377544,-1) -- (-2.598076211353316,-0.5);
        \draw [-,line width=0.5pt] (-2.598076211353316,-0.5) -- (-1.7320508075688772,0);
        \draw [-,line width=0.5pt] (-1.7320508075688772,0) -- (-0.8660254037844386,0.5);
        \draw [-,line width=0.5pt] (-1.7320508075688772,0) -- (-0.8660254037844386,-0.5);
        \draw [-,line width=0.5pt] (-2.598076211353316,-0.5) -- (-1.7320508075688772,-1);
        \draw [-,line width=0.5pt] (-1.7320508075688772,-1) -- (-0.8660254037844386,-0.5);
        \draw [-,line width=0.5pt] (-3.4641016151377544,-1) -- (-2.598076211353316,-1.5);
        \draw [-,line width=0.5pt] (-2.598076211353316,-1.5) -- (-1.7320508075688772,-1);
        \draw [-,line width=0.5pt] (-1.7320508075688772,-1) -- (-0.8660254037844386,-1.5);
        \draw [-,line width=0.5pt] (-2.598076211353316,-1.5) -- (-1.7320508075688772,-2);
        \draw [-,line width=0.5pt] (-1.7320508075688772,-2) -- (-0.8660254037844386,-1.5);
        \draw [-,line width=0.5pt] (-1.7320508075688772,-2) -- (-0.8660254037844386,-2.5);
        \begin{scriptsize}
        \draw [fill=black] (-3.4641016151377544,-1) circle (1.5pt);
        \draw[color=black] (-3.4762183455025912,-0.8860699289787186) node [yshift=1.5ex] {$\frac{88}{25}$};
        \draw [fill=black] (-2.598076211353316,-0.5) circle (1.5pt);
        \draw[color=black] (-2.5575235643431045,-0.41123891849179217) node [yshift=1.5ex]  {$\frac{184}{25}$};
        \draw [fill=black] (-1.7320508075688772,0) circle (1.5pt);
        \draw[color=black] (-1.6697960229979825,0.08939812517377155) node [yshift=1.5ex]  {$\frac{76}{5}$};
        \draw [fill=black] (-0.8660254037844386,0.5) circle (1.5pt);
        \draw[color=black] (-0.80271330819577,0.5900351688393352) node [yshift=1ex]  {$31$};
        \draw [fill=black] (-0.8660254037844386,-0.5) circle (1.5pt);
        \draw[color=black] (-0.8233581347386799,-0.41123891849179217) node [yshift=1ex]  {$7$};
        \draw [fill=black] (-1.7320508075688772,-1) circle (1.5pt);
        \draw[color=black] (-1.6904408495408922,-0.9118759621573559) node [yshift=1.5ex]  {$\frac{16}{5}$};
        \draw [fill=black] (-2.598076211353316,-1.5) circle (1.5pt);
        \draw[color=black] (-2.5575235643431045,-1.407351799187192) node [yshift=1.5ex]  {$\frac{36}{25}$};
        \draw [fill=black] (-0.8660254037844386,-1.5) circle (1.5pt);
        \draw[color=black] (-0.8233581347386799,-1.407351799187192) node [yshift=1ex]  {$1$};
        \draw [fill=black] (-1.7320508075688772,-2) circle (1.5pt);
        \draw[color=black] (-1.6904408495408922,-1.9079888428527558) node [yshift=1.5ex]  {$\frac{2}{5}$};
        \draw [fill=black] (-0.8660254037844386,-2.5) circle (1.5pt);
        \draw[color=black] (-0.6014262494023993,-2.382819853339682) node [xshift=-1.6ex] [yshift=1ex]  {$0$};
        \end{scriptsize}
    \end{tikzpicture}
  \caption{Discounted Option Values}
\end{subfigure}
\caption{Stock Prices and Option Value Trees}
\label{binomial}
\end{figure}

\end{Example}

\begin{remark}
Although \eqref{sequential} and \eqref{deltaHedging} give us equivalent results in the binomial case, it is important to note that \eqref{deltaHedging} is specifically limited to the binomial case, while \eqref{sequential} can be extended to general discrete-time market models, including incomplete models. In the following section, we apply \eqref{sequential} to the incomplete trinomial model.
\end{remark}
\section{Incomplete Markets}\label{secIncompleteMarkets}

While the binomial model is often used as an introductory tool, most models used in practice exhibit incompleteness, that is it is not always possible to construct a hedging strategy that  reduces loss to zero. We now introduce a tractable example of an incomplete market.

\subsection{Trinomial Asset Pricing Model}

 An example of an incomplete market is the \textit{trinomial asset pricing model}. Similar to the binomial model, we have a risky asset and a risk-free asset, and the value of the risky asset at each time step is determined by a small set of outcomes. This time, we have three possible outcomes for the coin flip instead of two. That is, along with the possibility of an increase by a factor of $u$ and decrease by a factor of $d$, we allow for the possibility that the stock price does not change between two consecutive time steps.

 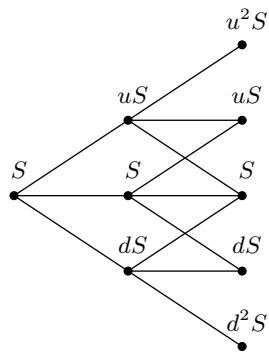
\begin{figure}[h]
 \centering
    \begin{tikzpicture}[line cap=round,line join=round,>=triangle 45,x=1cm,y=1cm]
    \draw [line width=0.5pt] (-5,0)-- (-3.5,1);
    \draw [line width=0.5pt] (-5,0)-- (-3.5,0);
    \draw [line width=0.5pt] (-5,0)-- (-3.5,-1);
    \draw [line width=0.5pt] (-3.5,1)-- (-2,2);
    \draw [line width=0.5pt] (-3.5,1)-- (-2,1);
    \draw [line width=0.5pt] (-3.5,0)-- (-2,1);
    \draw [line width=0.5pt] (-3.5,0)-- (-2,0);
    \draw [line width=0.5pt] (-3.5,1)-- (-2,0);
    \draw [line width=0.5pt] (-3.5,0)-- (-2,-1);
    \draw [line width=0.5pt] (-3.5,-1)-- (-2,0);
    \draw [line width=0.5pt] (-3.5,-1)-- (-2,-2);
    \draw [line width=0.5pt] (-3.5,-1)-- (-2,-1);
    \begin{scriptsize}
    \draw [fill=black] (-5,0) circle (1.5pt);
    \draw[color=black] (-4.934342944507471,0.17145307093634843) [yshift=5pt] node {$S$};
    \draw [fill=black] (-3.5,1) circle (1.5pt);
    \draw[color=black] (-3.4363801013803337,1.1674807205549196) [yshift=5pt] node {$uS$};
    \draw [fill=black] (-3.5,0) circle (1.5pt);
    \draw[color=black] (-3.4363801013803337,0.17145307093634843) [yshift=5pt] node {$S$};
    \draw [fill=black] (-3.5,-1) circle (1.5pt);
    \draw[color=black] (-3.4363801013803337,-0.8324173160807942) [yshift=5pt] node {$dS$};
    \draw [fill=black] (-2,2) circle (1.5pt);
    \draw[color=black] (-1.9384172582531956,2.1713511075720624) [yshift=5pt] node {$u^2S$};
    \draw [fill=black] (-2,1) circle (1.5pt);
    \draw[color=black] (-1.9384172582531956,1.1674807205549196) [yshift=5pt] node {$uS$};
    \draw [fill=black] (-2,0) circle (1.5pt);
    \draw[color=black] (-1.9384172582531956,0.17145307093634843) [yshift=5pt] node {$S$};
    \draw [fill=black] (-2,-1) circle (1.5pt);
    \draw[color=black] (-1.9384172582531956,-0.8324173160807942) [yshift=5pt] node {$dS$};
    \draw [fill=black] (-2,-2) circle (1.5pt);
    \draw[color=black] (-1.9384172582531956,-1.8284449656993653) [yshift=5pt] node {$d^2S$};
    \end{scriptsize}
    \end{tikzpicture}
    \caption{Example of a 2-period Trinomial Model}
    \label{trinomial}
\end{figure}
    \begin{remark}
        For the above formulation of the trinomial model, we again require that $u=\frac{1}{d}$. However, as in the binomial model, this is not necessarily the case.
    \end{remark}
Attempting backward recursion on this model using the non-discounted wealth process, we get 
\begin{displaymath}
    X_2(\w_1\w_2) = (1 + r)(X_1(\w_1) - \vth_1(\w_1)S_1(\w_1)) + \vartheta_1(\w_1)S_2(\w_1\w_2)
\end{displaymath}
and 
\begin{displaymath}
    X_1(\w_1) = (1 + r)(X_0 - \vartheta_0S_0) + \vartheta_0 S_1(\w_1) 
\end{displaymath}

Note that there are three possible values for $\w_1$ and three possible values for $\w_2$, and we must solve for $\vth_0$, $X_0$, each $\vartheta_1(\w_1)$, and each $X_1(\w_1)$, giving us eight (8) unknowns and twelve (12) equations. This makes the system overdetermined. In particular, simple matrix calculations reveal that, in general, we have no solution for $X_0$.

\section{Stability Under Model Perturbations}\label{secStability}
We now turn to the question of stability of the F\"ollmer-Schweizer decomposition. There are different kinds of perturbations one can consider, and for example, stability with respect to perturbations of $V_N$ is considered in \cite{monatStricker}. In this paper, we consider perturbations of the stock price process. For the stability analysis, as we work on finite probability spaces, the exact form of perturbations is not important, and we will suppose that there is a family of adapted stock price processes parametrized by $\varepsilon$, $\left( S^\varepsilon\right)_{\varepsilon\in (-\varepsilon_0,\varepsilon_0)}$, for some $\varepsilon_0 >0$. An example of such a family corresponds to linear perturbations of the drift and volatility considered in the following section. 
Here we will only suppose that 
\begin{displaymath}
S^\varepsilon \to S^0
\end{displaymath}
in the sense that 
\begin{equation}\label{Sstab}
\lim\limits_{\varepsilon\to 0}S^\varepsilon_n (\omega) = S^0_n (\omega),\quad \text{for every} \ n\in\{0,\dots,N\} \ \text{and} \ \omega\in\Omega.
\end{equation}
The following result asserts that the F\"ollmer-Schweizer decomposition on finite probability spaces is stable under perturbations of the stock of the form \eqref{Sstab}.
\begin{Theorem}\label{thmStab}
On a finite probability space, let us consider a family of stock price processes $\left( (S^\varepsilon_n)_{n\in\{0,\dots,N\}}\right)_{\varepsilon\in (-\varepsilon_0,\varepsilon_0)}$, for some $\varepsilon_0 >0$, satisfying \eqref{Sstab}. Let $V_N$ be given. Then the corresponding family of the F\"olmer-Schweizer decompositions
$$V_N = V_0^\varepsilon + \sum_{j=1}^N\widehat \vartheta^\varepsilon_j \Delta S^\varepsilon_j+L^\varepsilon_N,\quad \varepsilon\in (-\varepsilon_0,\varepsilon_0),$$
satisfies
\begin{equation}\label{FSstab}
\begin{split}
\lim\limits_{\varepsilon\to 0} V_0^\varepsilon &= V_0^0,\\
\lim\limits_{\varepsilon\to 0} L^\varepsilon_n &= L^0_n,\quad n\in\{0,\dots,N\},\\
\lim\limits_{\varepsilon\to 0} \widehat\vartheta^\varepsilon_n &= \widehat\vartheta^0_n,\quad n\in\{1,\dots,N\},
\end{split}
\end{equation} 
where the equalities hold for every $\omega\in\Omega$. 
As a consequence, we also have
\begin{equation}\label{Mstab}\lim\limits_{\varepsilon\to 0} 
 \sum_{j=1}^n\widehat \vartheta^\varepsilon_j \Delta S^\varepsilon_j =  \sum_{j=1}^n\widehat \vartheta^0_j \Delta S^0_j,\quad n\in\{1,\dots,N\},\ \omega\in\Omega.
\end{equation}
\end{Theorem}
\begin{proof}
The proof goes recursively, backward in $n$. First, let us consider $n = N$.  From \eqref{Sstab}, we get
\begin{displaymath}
\lim\limits_{\e\to 0}\Delta S_N^\e = S_N^0,\quad \omega\in\Omega.
\end{displaymath}
Further, we work on a finite probability space, via the formal definition of conditional expectation, from \eqref{Sstab}, without any additional assumptions,  we get 
\begin{displaymath}
 \lim\limits_{\varepsilon\to 0}\mathbb E_{\cF_{N-1}} \left[\Delta S_N^\e\right] = \mathbb E_{\cF_{N-1}} \left[\Delta S_N^0\right].
\end{displaymath}
As a consequence, we obtain
\begin{equation}\label{12241}
\begin{split}
   \lim\limits_{\varepsilon\to 0} \widehat \vartheta_{N}^\e &= \lim\limits_{\varepsilon\to 0} \frac{\Cov{\cF_{N-1}}{V_N}{\Delta S_N^\e}}{\Var{\cF_{N-1}}{\Delta S_N^\e}} \\
   & = \lim\limits_{\varepsilon\to 0} \frac{\mathbb E_{\cF_{N-1}}\left[\left(V_N -\mathbb E_{\cF_{N-1}} \left[V_N\right]\right) \left( \Delta S_N^\e -\mathbb E_{\cF_{N-1}} \left[\Delta S_N^\e\right]\right)\right]}{\mathbb E_{\cF_{N-1}}\left[ \left( \Delta S_N^\e -\mathbb E_{\cF_{N-1}} \left[\Delta S_N^\e\right]\right)^2\right]}\\
   & =  \frac{\mathbb E_{\cF_{N-1}}\left[\left(V_N -\mathbb E_{\cF_{N-1}} \left[V_N\right]\right) \lim\limits_{\varepsilon\to 0}\left( \Delta S_N^\e -\mathbb E_{\cF_{N-1}} \left[\Delta S_N^\e\right]\right)\right]}{\mathbb E_{\cF_{N-1}}\left[ \lim\limits_{\varepsilon\to 0}\left( \Delta S_N^\e -\mathbb E_{\cF_{N-1}} \left[\Delta S_N^\e\right]\right)^2\right]}\\
      & =  \frac{\mathbb E_{\cF_{N-1}}\left[\left(V_N-\mathbb E_{\cF_{N-1}} \left[V_N\right]\right)\left(  \lim\limits_{\varepsilon\to 0}\Delta S_N^\e - \lim\limits_{\varepsilon\to 0}\mathbb E_{\cF_{N-1}} \left[\Delta S_N^\e\right]\right)\right]}{\mathbb E_{\cF_{N-1}}\left[\left(  \lim\limits_{\varepsilon\to 0}\Delta S_N^\e - \lim\limits_{\varepsilon\to 0}\mathbb E_{\cF_{N-1}} \left[\Delta S_N^\e\right]\right)^2\right]} \\ 
            & =  \frac{\mathbb E_{\cF_{N-1}}\left[\left(V_N -\mathbb E_{\cF_{N-1}} \left[V_N\right]\right)\left( \Delta S_N^0 - \mathbb E_{\cF_{N-1}} \left[\Delta S_N^0\right]\right)\right]}{\mathbb E_{\cF_{N-1}}\left[\left( \Delta S_N^0 - \mathbb E_{\cF_{N-1}} \left[\Delta S_N^0\right]\right)^2\right]} \\ 
            &=\frac{\Cov{\cF_{N-1}}{V_N}{\Delta S_N^0}}{\Var{\cF_{N-1}}{\Delta S_N^0}} \\
            &=\widehat \vartheta_{N}^0,
   \end{split}
\end{equation}
where the chain of equalities holds for every $\omega\in\Omega$. 
If $N = 1$,\eqref{12241}  implies the third equality in \eqref{FSstab}, and the remaining assertions of the theorem follow.  If $N>1$
denoting $A^\e_{n} := V_N -  \sum\limits_{j = n + 1}^{N}\widehat \vartheta_j^\e\Delta S_j^\e
$, $n\in\{0,\dots,N-1\}$, $\e\in(-\e_0,\e_0)$, from \eqref{12241}, we get
$$  \lim\limits_{\varepsilon\to 0} A_{N-1}^\e = A_{N-1}^0,\quad\w\in\Omega.$$ Consequently, similarly to \eqref{12241}, we obtain
\begin{equation}\label{12242}
\begin{split}
\lim\limits_{\varepsilon\to 0} \widehat \vartheta_{N-1}^\e &= \lim\limits_{\varepsilon\to 0} \frac{\Cov{\cF_{N-2}}{A_{N-1}^\e}{\Delta S_{N-1}^\e}}{\Var{\cF_{N-2}}{\Delta S_{N-1}^\e}} \\
&= \frac{\Cov{\cF_{N-2}}{\lim\limits_{\varepsilon\to 0}A_{N-1}^\e}{\lim\limits_{\varepsilon\to 0}\Delta S_{N-1}^\e}}{\Var{\cF_{N-2}}{\lim\limits_{\varepsilon\to 0}\Delta S_{N-1}^\e}} \\
&= \frac{\Cov{\cF_{N-2}}{A_{N-1}^0}{\Delta S_{N-1}^0}}{\Var{\cF_{N-2}}{\Delta S_{N-1}^0}} \\
&= \widehat \vartheta_{N-1}^0,\quad \omega\in\Omega.
\end{split}
\end{equation}
Proceeding in such a manner, one can show that 
$$\lim\limits_{\varepsilon\to 0} \widehat\vartheta^\varepsilon_n = \widehat\vartheta^0_n,\quad n\in\{1,\dots,N\},\ \omega\in\Omega,$$
which is the last equality in \eqref{FSstab}. In turn, this and \eqref{Sstab} imply \eqref{Mstab}. Therefore, for every $\varepsilon\in (-\varepsilon_0,\varepsilon_0)$, by taking expectation in 
\begin{equation}\label{12243}V_N = V_0^\varepsilon + \sum_{j=1}^N\widehat \vartheta^\varepsilon_j \Delta S^\varepsilon_j+L^\varepsilon_N,
\end{equation} 
and using $\mathbb E\left[ L^\varepsilon_N\right] = 0$, $\varepsilon\in (-\varepsilon_0,\varepsilon_0)$, we deduce via \eqref{Mstab} that 
$\lim\limits_{\e\to 0}V^\e_0 = V^0_0$, i.e., the first equality in \eqref{FSstab} holds.  Consequently, as the left-hand side in \eqref{12243} does not depend on $\e$, from \eqref{12243}, the convergence of $V^\e_0$ to  $V^0_0$ and \eqref{Mstab}, we deduce that $\lim\limits_{\varepsilon\to 0} L^\e_N = L^0_N$ for every  $\w\in\Omega$. Finally, as $L^\e$'s are $\mathbb P$-martingales, using $\mathbb E_{\mathcal F_{n}}\left[L^\e_N \right] = L^\e_n$, we conclude that 
$$\lim\limits_{\e\to 0}L^\e_n =L^0_n,\quad \text{for every~} n\in\{0,\dots,N\}~\text{and}~\w\in\Omega,$$
which is the second equality in \eqref{FSstab}. The proof is complete. 
\end{proof}
\section{Asymptotic analysis}\label{secAsymptotics}
While stability tells us whether there is a convergence of the problem outputs under perturbations of the input data or not, the asymptotic analysis gives a quantitative estimate of how does the problem respond to such perturbations. In order to make such estimates assuming merely $S^\e\to S^0$ as in \eqref{Sstab} from the previous section is not enough. We need to parametrize perturbations more precisely. Before stating our form of perturbations, one can also consider that, in practice,  the dynamics of the stock price process, is commonly decomposed into two parts. The first one is drift, or in discounted settings, it can also be stated as the market price of risk. This part is responsible for the trend of the stock. The second part captures the fluctuations, that is, how much can the stock alternative in a given interval. Therefore, one can formulate the following dynamics of the stock for the base model\footnote{Here the base model is the one that corresponds to $\e= 0$ in the notations of sections \ref{secStability} and \ref{secAsymptotics}.}, 
        \begin{equation}\label{Ssem}
            \Delta S_n^0 = \lambda_n+ \sigma_n \Delta W_n = \lambda_n\Delta t+ \sigma_n \Delta W_n,\quad n\in\{1,\dots,N\},
        \end{equation}
        where $\lambda$ and $\sigma$ are predictable processes, $\Delta t \equiv 1$ represents the change in time, and $\Delta W_n$ is an $\cF_n$-measurable increment of a martingale with initial value $0$, interpreted as an error or ``noise" term, where we additionally suppose that the standard deviation of $\Delta W_n$ is $\Delta t=1$ for normalization purposes. \eqref{Ssem} is also consistent with the so-called semimartingale decomposition of the stock price process, where a semimartingale can be defined to be a process that can be written as a sum of a martingale (noise term) and an adapted process (drift term). In the finite probability settings, in fact, the drift term in the semimartingale decomposition can be chosen to be predictable. 
        \begin{Remark}\label{remLS}
        We observe that given any process $S$, and once the time step $\Delta t>0$ is fixed (and is constant in this paper, for simplicity of notations), the processes $\lambda$ and $\sigma$  can be obtained as follows
        \begin{displaymath}
        \begin{split}
        \lambda_n &= \frac{\mathbb E_{\cF_{n-1}}\left[ \Delta S_n\right]}{\Delta t},\\ 
        \sigma_n &= \sqrt{\Var{\cF_{n-1}}{\Delta S_n}},\\
         \Delta W_n &= \frac{\Delta S^0_n - \lambda_n\Delta t}{\sigma_n}1_{\{\sigma_n>0\}}, \quad  n\in\{1,\dots,N\}.
        \end{split}
        \end{displaymath}
        \end{Remark}

For perturbations of the underlying dynamics in \eqref{Ssem}, we can consider simultaneous (or separate) distortions of both the drift and the noise terms in \eqref{Ssem}. Therefore, we now define model perturbations as
        \begin{displaymath}
            \Delta S_n^\e = (\lambda_n + \e\lambda'_n)\Delta t + (\sigma_n + \e \sigma'_n)\Delta W_n + \e\sigma''_n\Delta W^{\perp}_n, \quad  n\in\{1,\dots,N\},\ \e\in(-\e_0,\e_0),
        \end{displaymath}
        where $\lambda'$, $\sigma'$, and $\sigma''$  are predictable processes, $\Delta W^{\perp}_n$ is an $\cF_n$-measurable normalized noise term, which is conditionally uncorrelated from $\Delta W_n$, that is $$\mathbb E_{\cF_{n-1}}\left[ \Delta W^{\perp}_n\right] = 0, \ \Var{\cF_{n-1}}{\Delta W^{\perp}_n} = 1, \ \text{and} \ \Cov{\cF_{n-1}}{\Delta { W}_n}{\Delta W^{\perp}_n} = 0,\ n\in\{1,\dots,N\}.$$ and $\e_0$ is a strictly positive constant. With 
        \begin{equation}\label{S'}\Delta S'_n:= \lambda'_n\Delta t + \sigma'_n\Delta W_n +\sigma''_n\Delta W^{\perp}_n,\quad  n\in\{1,\dots,N\}, \end{equation}
         we can rewrite the dynamics of the perturbed processes as
        \begin{equation}\label{SsemEps}
         \Delta S_n^\e = \Delta S_n^0 + \e\Delta S'_n,\quad  n\in\{1,\dots,N\}, \ \e\in(-\e_0,\e_0).
        \end{equation}
        \begin{Remark}
        Similarly to the process $\Delta S$, if one starts from a given perturbation process $\Delta{S'}$, it can be represented in the form \eqref{S'}, as follows: in the notations of Remark \ref{remLS}, we have
        \begin{displaymath}
        \lambda'_n = \frac{\mathbb E_{\cF_{n-1}}\left[ \Delta S'_n\right]}{\Delta t},\quad  n\in\{1,\dots,N\},
        \end{displaymath}
        and the computation of the remaining parameters goes along the lines of \cite[Example 2.3.3, p.72]{Shreve2}. Let us consider 
        \begin{displaymath}
       \Delta S'_n - \mathbb E_{\cF_{n-1}}\left[ \Delta S'_n\right] = 
       \sigma'_n \Delta W_n + \left( \Delta S'_n - \mathbb E_{\cF_{n-1}}\left[ \Delta S'_n\right] -\sigma'_n  \Delta W_n \right),\quad  n\in\{1,\dots,N\},
        \end{displaymath}
        where $\sigma'_n$ has to be determined in a way that term in the brackets is conditionally uncorrelated from $\Delta W_n$. As, $\Var{\cF_{n-1}}{\Delta W_n} =1$, direct calculations give
        \begin{displaymath}
        \sigma'_n = \mathbb E_{\cF_{n-1}}\left[ \Delta W_n,  \Delta S'_n - \mathbb E_{\cF_{n-1}}\left[ \Delta S'_n\right] \right]
        =\Cov{\cF_{n-1}}{\Delta W_n}{\Delta S'_n},\quad  n\in\{1,\dots,N\},
        \end{displaymath}
        therefore, for $\sigma''_n$ and $\Delta W_n^{\perp}$, we get
        \begin{displaymath}\begin{split}
        \sigma''_n &= \sqrt{\Var{\cF_{n-1}}{\Delta {S}'_n - \mathbb E_{\cF_{n-1}}\left[{\Delta S'_n}\right] -\sigma'_n  \Delta W_n }}
        = \sqrt{\Var{\cF_{n-1}}{{\Delta S'_n}  -\sigma'_n  \Delta W_n }},\\ 
        \Delta W_n^{\perp} &= \frac{\Delta {S}'_n - \mathbb E_{\cF_{n-1}}\left[{\Delta S'_n}\right] -\sigma'_n  \Delta W_n}{\sigma''_n}1_{\{\sigma''_n>0\}},\quad  n\in\{1,\dots,N\}.
        \end{split}
        \end{displaymath}
          \end{Remark}
        The following theorem gives the leading-order correction terms to the components of the F\"ollmer-Schweizer decomposition under perturbations of the form \eqref{SsemEps}. Let us consider the associated family of the F\"olmer-Schweizer decompositions
\begin{equation}\label{FSfam}
V_N = V_0^\varepsilon + \sum_{j=1}^N\widehat \vartheta^\varepsilon_j \Delta S^\varepsilon_j+L^\varepsilon_N,\quad \varepsilon\in (-\varepsilon_0,\varepsilon_0),
\end{equation}
        and let us define recursively, backward in $n$,  process $\widehat\vartheta'$, which will be proven in Theorem \ref{thmAsym} to be the first-order correction to the optimal strategy, as 
\begin{equation}\label{theta'}
\begin{split}
\widehat\vartheta'_N &:= \frac{
\Cov{\cF_{N-1}}{V_N}{\Delta S'_N}}{\Var{\cF_{N-1}}{\Delta S^0_N}}-2\frac{\Cov{\cF_{N-1}}{V_N}{\Delta S^0_N}\Cov{\cF_{N-1}}{\Delta S^0_N}{\Delta S'_N}
}{\left(\Var{\cF_{N-1}}{\Delta S^0_N}\right)^2},\\
\widehat\vartheta'_n& :=\frac{1}{\Var{\cF_{n-1}}{\Delta S^0_n}}\left(\Cov{\cF_{n-1}}{V_N -  \sum\limits_{j = n + 1}^{N}\widehat \vartheta_j^0\Delta S_j^0 }{\Delta S'_n} 
\right. \\
&\hspace{35mm}\left.- \Cov{\cF_{n-1}}{\sum\limits_{j = n + 1}^{N} \widehat\vartheta'_j\Delta S_j^0 + \sum\limits_{j = n + 1}^{N}\widehat \vartheta_j^0\Delta S'_j}{\Delta S^0_n}\right)\\
&\qquad-2\frac{\Cov{\cF_{n-1}}{V_N -  \sum\limits_{j = n + 1}^{N}\widehat \vartheta_j^0\Delta S_j^0 }{\Delta S^0_n} \Cov{\cF_{N-1}}{\Delta S^0_n}{\Delta S'_n}}{\left(\Var{\cF_{n-1}}{\Delta S^0_n}\right)^2},\\
&\hspace{85mm} n \in\{N-1, \dots,1\}. 
\end{split}
\end{equation}

        \begin{Theorem}\label{thmAsym}
On a finite probability space, let us consider a family of stock price processes $\left( (S^\varepsilon_n)_{n\in\{0,\dots,N\}}\right)_{\varepsilon\in (-\varepsilon_0,\varepsilon_0)}$, for some $\varepsilon_0 >0$, where the increments, ${\Delta S_n^\e}'$s, are given via  \eqref{SsemEps}. Let $V_N$ be given. 
Then the components of the family of the F\"ollmer-Schweizer decompositions defined in \eqref{FSfam} satisfy
\begin{equation}\label{FSasym}
\begin{split}
\lim\limits_{\varepsilon\to 0} \frac{V_0^\varepsilon - V_0^0}{\e} &=  -\mathbb E\left[ \sum_{j=1}^N\widehat \vartheta^0_j \Delta S'_j +\sum_{j=1}^N\widehat \vartheta'_j \Delta S^0_j\right],\\
\lim\limits_{\varepsilon\to 0} \frac{L^\varepsilon_n - L^0_n}{\e} &= -\mathbb E_{\cF_{n}}\left[\sum\limits_{j=1}^N \widehat\vartheta'_j\Delta S^0_j - \mathbb E\left[\sum\limits_{j=1}^N \widehat\vartheta'_j\Delta S^0_j\right]\right]
-\mathbb E_{\cF_{n}}\left[\sum\limits_{j=1}^N \widehat\vartheta^0_j\Delta S'_j - \mathbb E\left[\sum\limits_{j=1}^N \widehat\vartheta^0_j\Delta S'_j\right]\right],\\
&\hspace{100mm} n\in\{0,\dots,N\},\ \omega\in\Omega,\\
\lim\limits_{\varepsilon\to 0} \frac{\widehat\vartheta^\varepsilon_n - \widehat\vartheta^0_n}{\e} &=\widehat\vartheta'_n,\quad n\in\{1,\dots,N\},\ \omega\in\Omega,
\end{split}
\end{equation}  
where $\vartheta'$ is defined in \eqref{theta'}. 
We also have
\begin{equation}\label{Masym}
\begin{split}
\lim\limits_{\varepsilon\to 0} 
\frac{
 \sum_{j=1}^n\widehat \vartheta^\varepsilon_j \Delta S^\varepsilon_j - \sum_{j=1}^n\widehat \vartheta^0_j \Delta S^0_j}{\e} 
 = \sum_{j=1}^n\widehat \vartheta^0_j \Delta S'_j +\sum_{j=1}^n\widehat \vartheta'_j \Delta S^0_j,\\
  n\in\{1,\dots,N\},\omega\in\Omega.
  \end{split}
\end{equation}
\end{Theorem}
\begin{proof}
We will investigate $\lim\limits_{\varepsilon\to 0} \frac{\widehat\vartheta^\varepsilon_n - \widehat\vartheta^0_n}{\e} $ first.
For $n=N$, we have
\begin{equation}\label{12255}
\begin{split}
\lim\limits_{\varepsilon\to 0} \frac{\widehat\vartheta^\varepsilon_n - \widehat\vartheta^0_n}{\e} = \lim\limits_{\varepsilon\to 0} \frac{\frac{\Cov{\cF_{N-1}}{V_N}{\Delta S_N^\e}}{\Var{\cF_{N-1}}{\Delta S_N^\e}} - \frac{\Cov{\cF_{N-1}}{V_N}{\Delta S_N^\e}}{\Var{\cF_{N-1}}{\Delta S_N^\e}}}{\e}.
\end{split}
\end{equation}
Using the definition of conditional expectation, and Theorem \ref{thmStab}\footnote{Below, we apply the assertions of Theorem \ref{thmStab} at a family of points $\e$ near the origin, this however causes no difficulty by relabeling $S^\e$'s.}, we get 
$$\lim\limits_{\e\to 0}\frac{
\Cov{\cF_{N-1}}{V_n}{\Delta S^\e_N}-
\Cov{\cF_{N-1}}{V_n}{\Delta S^0_N}
}
{\e} = \Cov{\cF_{N-1}}{V_N}{\Delta S'_N}.$$
Similarly, we deduce that
$$\lim\limits_{\e\to 0}\frac{
\Var{\cF_{N-1}}{\Delta S^\e_N} -\Var{\cF_{N-1}}{\Delta S^0_N} 
}
{\e}
= 2\Cov{\cF_{N-1}}{\Delta S^0_N}{\Delta S'_N}.$$
Therefore, in \eqref{12255}, we obtain
\begin{displaymath}
\begin{split}
&\lim\limits_{\varepsilon\to 0} \frac{\widehat\vartheta^\varepsilon_n - \widehat\vartheta^0_n}{\e}=\\ 
&\frac{
\Cov{\cF_{N-1}}{V_N}{\Delta S'_N}\Var{\cF_{N-1}}{\Delta S^0_N} -2\Cov{\cF_{N-1}}{V_N}{\Delta S^0_N}\Cov{\cF_{N-1}}{\Delta S^0_N}{\Delta S'_N}
}{\left(\Var{\cF_{N-1}}{\Delta S^0_N}\right)^2},
\end{split}
\end{displaymath}
which is precisely $\widehat\vartheta'_N$.  
If $N=1$, this completes the proof for $\lim\limits_{\varepsilon\to 0} \frac{\widehat\vartheta^\varepsilon_n - \widehat\vartheta^0_n}{\e} = \vartheta'_n$, for every $n$. If $N>1$
denoting $A^\e_{n} := V_N -  \sum\limits_{j = n + 1}^{N}\widehat \vartheta_j^\e\Delta S_j^\e
$, $n\in\{0,\dots,N-1\}$, $\e\in(-\e_0,\e_0)$, and using Theorem \ref{thmStab}, we get recursively, backward in $n$, the following chain of equalities. First, for $n=N-1$, we obtain
\begin{displaymath}
\lim\limits_{\e \to 0}\frac{A^\e_n - A^0_n}{\e} = 
- \sum\limits_{j = n + 1}^{N} \widehat\vartheta'_j\Delta S_j^0 - \sum\limits_{j = n + 1}^{N}\widehat \vartheta_j^0\Delta S'_j.
\end{displaymath}
Therefore, using Theorem \ref{thmStab} again, we obtain
\begin{displaymath}
\begin{split}
&\lim\limits_{\e \to 0}\frac{
\Cov{\cF_{n-1}}{A^\e_n }{\Delta S^\e_n}- \Cov{\cF_{n-1}}{A^0_n }{\Delta S^0_n}
}{\e} \\
&=\Cov{\cF_{n-1}}{- \sum\limits_{j = n + 1}^{N} \widehat\vartheta'_j\Delta S_j^0 - \sum\limits_{j = n + 1}^{N}\widehat \vartheta_j^0\Delta S'_j}{\Delta S^0_n} + \Cov{\cF_{n-1}}{A^0_n }{\Delta S'_n}\\
& = \Cov{\cF_{n-1}}{- \sum\limits_{j = n + 1}^{N} \widehat\vartheta'_j\Delta S_j^0 - \sum\limits_{j = n + 1}^{N}\widehat \vartheta_j^0\Delta S'_j}{\Delta S^0_n} + \Cov{\cF_{n-1}}{V_N -  \sum\limits_{j = n + 1}^{N}\widehat \vartheta_j^0\Delta S_j^0 }{\Delta S'_n},
\end{split}
\end{displaymath}
and thus, in \eqref{12255}, we conclude that
\begin{equation}\nonumber
\begin{split}
\lim\limits_{\varepsilon\to 0} \frac{\widehat\vartheta^\varepsilon_n - \widehat\vartheta^0_n}{\e} =&\frac{1}{\Var{\cF_{n-1}}{\Delta S^0_n}}\left(\Cov{\cF_{n-1}}{V_N -  \sum\limits_{j = n + 1}^{N}\widehat \vartheta_j^0\Delta S_j^0 }{\Delta S'_n} 
\right. \\
&\hspace{30mm}\left.- \Cov{\cF_{n-1}}{\sum\limits_{j = n + 1}^{N} \widehat\vartheta'_j\Delta S_j^0 + \sum\limits_{j = n + 1}^{N}\widehat \vartheta_j^0\Delta S'_j}{\Delta S^0_n}\right)\\
&\quad-\frac{2\Cov{\cF_{n-1}}{V_N -  \sum\limits_{j = n + 1}^{N}\widehat \vartheta_j^0\Delta S_j^0 }{\Delta S^0_n} \Cov{\cF_{N-1}}{\Delta S^0_n}{\Delta S'_n}}{\left(\Var{\cF_{n-1}}{\Delta S^0_n}\right)^2},
\end{split}
\end{equation}
which is exactly $\widehat\vartheta'_n$ from \eqref{theta'} for $n = N-1$. Proceeding this way, we  can establish \eqref{theta'} for every $n\geq 1$. This completes the proof of the third equality in \eqref{FSasym}. Now, \eqref{Masym}  follows from the third equality in \eqref{FSasym} and Theorem \ref{thmStab}.

To establish the first two equalities in \eqref{FSasym}, we proceed as follows. For every $\e\in(-\e_0,\e_0)$, taking the expectation in \eqref{FSfam} 
and observing that the left-hand side does not depend on $\varepsilon$ as well as that $\mathbb E\left[ L^\varepsilon_N\right] = 0$, we get
$$V_0^\varepsilon + \mathbb E\left[\sum_{j=1}^N\widehat \vartheta^\varepsilon_j \Delta S^\varepsilon_j \right]= V_0^0 + \mathbb E\left[\sum_{j=1}^N\widehat \vartheta^0_j \Delta S^0_j \right],\quad \e\in(-\e_0,\e_0).$$
Collecting the terms and dividing by $\e\neq 0$, we obtain
$$\frac{V_0^\varepsilon - V_0^0}{\e} = 
\frac{
 \mathbb E\left[\sum_{j=1}^N\widehat \vartheta^0_j \Delta S^0_j -\sum_{j=1}^N\widehat \vartheta^\varepsilon_j \Delta S^\varepsilon_j \right]
 }
 {\e},\quad \e\in(-\e_0,0)\cup(0,\e_0).
 $$
 Taking the limit as $\e\to 0$, and using \eqref{Masym}, we conclude that 
 $$\lim\limits_{\e\to 0}\frac{V_0^\varepsilon - V_0^0}{\e} = -\mathbb E\left[ \sum_{j=1}^N\widehat \vartheta^0_j \Delta S'_j +\sum_{j=1}^N\widehat \vartheta'_j \Delta S^0_j\right],$$
 which is precisely the first equality in \eqref{FSasym}.
 
 To obtain the remaining assertion in \eqref{FSasym}, we observe that from \eqref{FSfam}, \eqref{Masym} and the other two assertions in \eqref{FSasym}, we immediately obtain 
\begin{displaymath}\begin{split}\lim\limits_{\e\to 0}\frac{L^\e_N - L^0_N}{\e} &= -\lim\limits_{\e\to 0}\frac{V^\e_0 - V^0_0}{\e} - \lim\limits_{\e\to 0}\frac{\sum_{j=1}^N\widehat \vartheta^\varepsilon_j \Delta S^\varepsilon_j - \sum_{j=1}^N\widehat \vartheta^0_j \Delta S^0_j}{\e} \\ 
&= \mathbb E\left[ \sum_{j=1}^N\widehat \vartheta^0_j \Delta S'_j +\sum_{j=1}^N\widehat \vartheta'_j \Delta S^0_j\right] - \sum_{j=1}^N\widehat \vartheta^0_j \Delta S'_j -\sum_{j=1}^N\widehat \vartheta'_j \Delta S^0_j\\ 
&= -\left(\sum_{j=1}^N\widehat \vartheta^0_j \Delta S'_j -\mathbb E\left[ \sum_{j=1}^N\widehat \vartheta^0_j \Delta S'_j \right]\right) 
-\left(\sum_{j=1}^N\widehat \vartheta'_j \Delta S^0_j - \mathbb E\left[ \sum_{j=1}^N\widehat \vartheta'_j \Delta S^0_j\right]\right).
\end{split}
\end{displaymath}
If $N=1$, this completes the proof. 
If $N>1$, for $n<N$, we have $\mathbb E_{\cF_{n}}\left[ L^\e_N\right] = L^\e_n$, $\e\in(-\e_0,\e_0).$ Therefore, using \eqref{FSfam} and taking the conditional expectation, we obtain
$$\frac{L^\e_n - L^0_n}{\e} = \frac{\mathbb E_{\cF_{n}}\left[ L^\e_N - L^0_N\right] }{\e}
=-\frac{V^\e_0 - V^0_0}{\e} - \frac{\mathbb E_{\cF_{n}}\left[
\sum\limits_{j=1}^N\widehat\vartheta^\e_j\Delta S^\e_j-  \sum\limits_{j=1}^N\widehat\vartheta^0_j\Delta S^0_j \right]}{\e}
.$$
Taking the limit, and using \eqref{Masym} and the first equality in \eqref{FSasym}, we conclude that 
\begin{displaymath}
\lim\limits_{\e\to 0}\frac{L^\e_n - L^0_n}{\e}  = 
-\mathbb E_{\cF_{n}}\left[\sum\limits_{j=1}^N \widehat\vartheta'_j\Delta S^0_j - \mathbb E\left[\sum\limits_{j=1}^N \widehat\vartheta'_j\Delta S^0_j\right]\right]
-\mathbb E_{\cF_{n}}\left[\sum\limits_{j=1}^N \widehat\vartheta^0_j\Delta S'_j - \mathbb E\left[\sum\limits_{j=1}^N \widehat\vartheta^0_j\Delta S'_j\right]\right].
\end{displaymath}
This completes the proof of the theorem.
\end{proof}

\bibliographystyle{alpha}
\bibliography{referencesDec2019}
\end{document}